\documentclass[aps,pra,twocolumn,showpacs,nofootinbib,longbibliography,floatfix,superscriptaddress]{revtex4-1}

\usepackage{amsmath,amssymb,amsfonts,amsthm}
\usepackage{mathtools}
\usepackage{bm}
\usepackage{graphicx}
\usepackage{xcolor}
\usepackage{braket}
\usepackage{hyperref}
\hypersetup{colorlinks=true,linkcolor=blue!50!black,citecolor=blue!50!black,urlcolor=blue!50!black}
\usepackage[normalem]{ulem}

\newtheorem{theorem}{Theorem}
\newtheorem{proposition}[theorem]{Proposition}

\newcommand{\GHZ}{\mathrm{GHZ}}
\newcommand{\W}{\mathrm{W}}
\newcommand{\nv}{\vec{n}}

\newcommand{\sgma}{\bm{\sigma}}
\newcommand{\xh}{\hat{x}}
\newcommand{\yh}{\hat{y}}
\newcommand{\zh}{\hat{z}}

\newcommand{\eopV}[3]{\langle\sigma(\nv_{#1},\nv_{#2},\nv_{#3})\rangle}
\newcommand{\E}[3]{E_{#1#2#3}}                    

\newcommand{\Eghz}{\mathcal{E}^{\,\mathrm{al}}_{\GHZ}}   
\newcommand{\Elu}{\mathcal{E}_{\GHZ}}                    
\newcommand{\Igen}{I_{\star}}                            
\newcommand{\tr}{\mathrm{Tr}}

\begin{document}

\title{A measure for genuine tripartite entanglement}

\author{Shengjun Wu$^{1,2}$, Kaichen Zhong$^{1,2}$, and Jeffery Wu$^{3}$ \\
{\small\itshape $^1$National Laboratory of Solid State Microstructures and School of Physics,} \\
{\small\itshape Collaborative Innovation Center of Advanced Microstructures, Nanjing University, Nanjing 210093, China}\\
{\small\itshape $^2$Kuang Yaming Honors School, Nanjing University, Nanjing 210023, China}\\
{\small\itshape $^3$School of Computer Science and School of Physics and Astronomy, Shanghai Jiao Tong University, Shanghai 200240, China}\\[2pt]
}

\date{\today}

\begin{abstract}
We introduce a real-valued functional $I(\nv_1,\nv_2)$ of four
three-qubit correlations that turns the Greenberger-Horne-Zeilinger
(GHZ) algebraic paradox into a quantitative witness of genuine
tripartite entanglement.  For every three-qubit state $\rho$ and every
pair of directions, $|I(\nv_1,\nv_2;\rho)|\le2$, with equality if and
only if $\nv_1\perp\nv_2$ and $\rho$ is locally unitarily equivalent to
the GHZ state.  On the five-parameter Ac\'in family, $I(\xh,\yh)$
depends only on $\lambda_0\lambda_4$ and is maximised by the GHZ state;
for the W state $I(\xh,\yh)=0$ and $\max|I_\W|=35/27$.  Maximising over
independent local orthonormal frames yields a local-unitary invariant
measure $\Elu\in[0,1]$, computable from four three-body Pauli
correlators, that equals one exactly on the GHZ orbit and takes the
values $35/54$ on the W state and $1/2$ on product states.  We show that
$\Elu$ is not convex: an equal mixture of three biseparable states has
$\Elu=14/27$, larger than the average $1/2$ of its components.  We also
relate $\Elu$ to GHZ fidelity witnesses, discuss its behaviour under
local operations and classical communication, and outline a
generalisation to three qudits via Heisenberg-Weyl operators.
\end{abstract}

\maketitle

\section{Introduction}\label{sec:motivation}

The Greenberger--Horne--Zeilinger (GHZ) theorem~\cite{ghz1989,ghsz1990}
shows that quantum mechanics conflicts with local realism in a
\emph{deterministic} way for three or more parties, requiring no
inequalities and no averages over runs.  A single algebraic identity
relating four products of dichotomic outcomes is enough to produce a
contradiction~\cite{mermin1990,bk1993,ardehali1992}.  Although the
qualitative content of this identity is celebrated, the quantitative
question \emph{how strongly does a given state $\rho$ exhibit
GHZ-type correlations?} has remained somewhat scattered across the
literature, with answers expressed through Mermin polynomials, fidelities
with $\ket{\GHZ}$, the three-tangle~\cite{ckw2000}, or convex-roof
extensions of various entanglement monotones~\cite{horodecki2009}.
Quantitative studies in concrete physical settings, for instance the
evolution of the entanglement of three-particle entangled states in the
Tavis--Cummings model~\cite{zuo2003}, or entropic entanglement measures
for two-qubit mixed states based on the negativity of the quantum
conditional entropy~\cite{zhou2007}, illustrate how strongly the answer
depends on the functional one chooses.  The
purpose of this paper is to formulate and analyse a single functional
$I(\nv_1,\nv_2)$ that
\begin{enumerate}
\item[(i)] exactly reproduces the GHZ algebraic paradox when
  $\nv_1=\xh$, $\nv_2=\yh$;
\item[(ii)] satisfies a tight upper bound, $|I|\le 2$, that is universal
  across all three-qubit states;
\item[(iii)] saturates this bound \emph{only} on the GHZ orbit under
  local unitaries, thereby providing a witness of GHZ-type genuine
  tripartite entanglement;
\item[(iv)] admits closed-form evaluation on canonical state
  families (Ac\'in form, W class, biseparable states); and
\item[(v)] extends naturally from qubits to qudits.
\end{enumerate}

\textit{Position with respect to prior work.}
Our quantity sits at the intersection of three lines of research.
First, in contrast to the standard Mermin
polynomial~\cite{mermin1990,ardehali1992,bk1993,werner2001} and its
genuine-tripartite refinement by Svetlichny and
Seevinck~\cite{svetlichny1987,seevinck2002}, the functional $I$ is
\emph{multiplicative} in the correlators rather than linear; the
algebraic identity~\eqref{eq:LHV-identity} replaces the inequality.
Multiplicative Bell-type quantities have been considered abstractly by
Te'eni \emph{et al.}~\cite{Teeni2019} and (in quadratic form) by
Uffink~\cite{uffink2002}, but never specialised to the GHZ stabiliser.
Second, in contrast to the inequality-free GHZ-vs-W paradoxes of
Cabello~\cite{cabello2002}, our construction yields a single LU-invariant
\emph{scalar} that orders all three-qubit states on a $[0,1]$ scale.
Third, in contrast to the three-tangle $\tau_3$~\cite{ckw2000}, GHZ
fidelity witnesses~\cite{toth2005,bourennane2004}, the genuine
multipartite negativity~\cite{hofmann2014}, and the GME concurrence of
Ma \emph{et al.}~\cite{ma2011}, our $\Elu$ is non-zero on \emph{both}
the W class and biseparable states, with sharp ordering
$\Elu(\GHZ)=1>\Elu(\W)=35/54> 1/2 \ge \Elu(\textnormal{bisep})$,
where ``bisep'' refers to states that are separable across a
\emph{fixed} bipartition.  The bound extends to convex mixtures across
any \emph{two} bipartitions~\cite{ZhangMa2026} (Sec.~\ref{sec:twocut}) but not to mixtures
across all three (Sec.~\ref{sec:counterexample}).

\textit{Notation.}
Three qubits live in a Hilbert space $\mathcal{H}=\mathbb{C}^2_A\otimes\mathbb{C}^2_B
\otimes\mathbb{C}^2_C$, with state $\rho^{ABC} \in \mathcal{S}(\mathcal{H})$.  Marginals are
$\rho^{AB}=\tr_C\rho^{ABC}$, etc.  For a unit vector
$\nv\in\mathbb{R}^3$, the spin observable on a single qubit is
$\sigma_{\nv}=\nv\cdot\sgma=n_x\sigma_x+n_y\sigma_y+n_z\sigma_z$, with
$\sigma_x,\sigma_y,\sigma_z$ the standard Pauli matrices.  For three
direction labels $\nv_a,\nv_b,\nv_c$ we denote the tripartite local
observable
\begin{equation}
\sigma(\nv_a,\nv_b,\nv_c)\;=\;\sigma_{\nv_a}\otimes\sigma_{\nv_b}\otimes\sigma_{\nv_c},
\end{equation}
and its expectation
$\langle\sigma(\nv_a,\nv_b,\nv_c)\rangle
=\tr[\sigma(\nv_a,\nv_b,\nv_c)\rho^{ABC}]$.

\textit{The functional.}
For two unit vectors $\nv_1,\nv_2\in\mathbb{R}^3$ we define, using the
short-hand $\E{a}{b}{c}\equiv\eopV{a}{b}{c}$,
\begin{equation}\label{eq:Idef}
I(\nv_1,\nv_2; \rho) = \E{1}{1}{1} - \E{1}{2}{2} \E{2}{1}{2} \E{2}{2}{1} .
\end{equation}
The quantity $I$ is a real-valued function of the state $\rho$ and the
two directions; only product correlations of the spin observables on
each subsystem appear, so $I$ is, in principle, accessible to
estimation from purely
local Pauli-string measurements, with no shared reference frame
between the three parties (see Sec.~\ref{sec:max} for what this does
and does not certify).

\textit{Local hidden-variable identity.}
For any deterministic local hidden-variable (LHV) model assigning values
$A(\nv),B(\nv),C(\nv)\in\{\pm1\}$ to the three qubits, every realisation
of the underlying parameter $\lambda$ obeys the algebraic identity
\begin{multline}\label{eq:LHV-identity}
\bigl[A(\nv_1)B(\nv_2)C(\nv_2)\bigr]\bigl[A(\nv_2)B(\nv_1)C(\nv_2)\bigr]\\
\times\bigl[A(\nv_2)B(\nv_2)C(\nv_1)\bigr]=A(\nv_1)B(\nv_1)C(\nv_1),
\end{multline}
because each of $A(\nv_2)^2,B(\nv_2)^2,C(\nv_2)^2$ equals $+1$.  Hence
\emph{whenever the three product-correlation factors in $I$ are
simultaneously deterministic at the LHV level}, eq.~\eqref{eq:Idef}
forces
\begin{equation}\label{eq:LHV-zero}
I(\nv_1,\nv_2)\,\big|_{\textnormal{deterministic LHV}}=0,
\end{equation}
which is the standard GHZ--Mermin equality~\cite{mermin1990,ghsz1990}.
For the special case $\nv_1=\xh$, $\nv_2=\yh$ on the GHZ state
$\ket{\GHZ}=(\ket{000}+\ket{111})/\sqrt{2}$, quantum mechanics gives
$\eopV{1}{2}{2}=\eopV{2}{1}{2}=\eopV{2}{2}{1}=-1$ and
$\eopV{1}{1}{1}=+1$, so
\begin{equation}\label{eq:Ighzminus2}
I_{\GHZ}(\xh,\yh)=1- (-1)^3= 2,
\end{equation}
manifestly violating~\eqref{eq:LHV-zero}.  The maximal violation
$|I|=2$ is the algebraic essence of the GHZ paradox.

The remainder of this paper studies the functional $I$
quantitatively: who else can saturate the bound (Sec.~\ref{sec:max}),
what value it takes on a generic three-qubit pure state
(Sec.~\ref{sec:eval}), how it gives rise to a genuine tripartite
entanglement measure (Sec.~\ref{sec:measure}), and how it lifts to
qudits (Sec.~\ref{sec:qudit}).

\section{The four stabiliser-like operators}\label{sec:operators}

It is convenient to package $I(\nv_1,\nv_2)$ in terms of four Hermitian
observables on $\mathcal{H}$:
\begin{equation}\label{eq:OOOO}
\begin{aligned}
O_1 &= \sigma_{\nv_1}\otimes\sigma_{\nv_2}\otimes\sigma_{\nv_2},\\
O_2 &= \sigma_{\nv_2}\otimes\sigma_{\nv_1}\otimes\sigma_{\nv_2},\\
O_3 &= \sigma_{\nv_2}\otimes\sigma_{\nv_2}\otimes\sigma_{\nv_1},\\
O_4 &= \sigma_{\nv_1}\otimes\sigma_{\nv_1}\otimes\sigma_{\nv_1}.
\end{aligned}
\end{equation}
With $e_i\equiv\langle O_i\rangle_\rho$ ($i=1,\dots,4$),
\begin{equation}
I(\nv_1,\nv_2)\;=\;e_4 - e_1 e_2 e_3 .
\end{equation}
Each $O_i$ is unitary and Hermitian and obeys $O_i^2=\openone$, hence
$|e_i|\le 1$ and $|I|\le 2$ trivially.

The algebraic structure of these operators is most transparent in terms
of the inner and cross products of the two direction vectors,
$c\equiv\nv_1\cdot\nv_2$ and $\bm{m}\equiv\nv_1\times\nv_2$.  Using
$\sigma_{\nv_1}\sigma_{\nv_2}=c\,\openone+\mathrm{i}\,\bm{m}\cdot\sgma$
one finds the following two key identities (proved in
App.~\ref{app:identities}):
\begin{align}
[O_1,O_2] &= 2\mathrm{i}\,c\,
   \bigl(\bm{m}\cdot\sgma\!\otimes\!\openone-\openone\!\otimes\!\bm{m}\cdot\sgma\bigr)
   \otimes\openone,\label{eq:O1O2-comm}\\
O_1 O_2 O_3 + O_4
   &= 2c\,\sigma_{\nv_1}\otimes\sigma_{\nv_2}\otimes\sigma_{\nv_1}.
   \label{eq:O1O2O3-O4}
\end{align}
Equation~\eqref{eq:O1O2O3-O4} is the operator avatar of
Mermin's identity~\eqref{eq:LHV-identity}: the right-hand side vanishes
exactly when $\nv_1\perp\nv_2$, in which case
$\{O_1,O_2,O_3,O_4\}$ commute pairwise and satisfy the abelian
relation
\begin{equation}\label{eq:GHZ-stabiliser}
O_1\,O_2\,O_3\,O_4=-\openone\qquad(\text{when }\nv_1\perp\nv_2).
\end{equation}
This is precisely the defining relation of the three-qubit GHZ
stabiliser group~\cite{gottesman1997,nielsen2010} written in the
locally rotated Pauli basis $(\sigma_{\nv_1},\sigma_{\nv_2})$.

\section{Maximum violation: a sharp bound}\label{sec:max}

Our first main result completely characterises the states and
directions saturating $|I|=2$.

\begin{proposition}[Maximum violation]\label{thm:max}
For every three-qubit state $\rho$ and every pair of unit vectors
$\nv_1,\nv_2\in\mathbb{R}^3$ one has
\begin{equation}\label{eq:Ibound}
|I(\nv_1,\nv_2;\rho)|\;\le\;2.
\end{equation}
Equality holds if and only if
\begin{enumerate}
\item[\textnormal{(i)}] $\nv_1\perp\nv_2$, and
\item[\textnormal{(ii)}] $\rho=\ket{\psi}\bra{\psi}$ is pure and a
common eigenstate of the four operators
$\{O_1,O_2,O_3,O_4\}$ with respective eigenvalues
$\epsilon_i=\pm1$ satisfying
$\epsilon_1\epsilon_2\epsilon_3=-\epsilon_4$.
\end{enumerate}
The eight such common eigenstates form an orthonormal basis of
$\mathcal{H}$, all of whose elements are connected to the standard GHZ
state $\ket{\GHZ}=(\ket{000}+\ket{111})/\sqrt{2}$ by local unitary
transformations.
\end{proposition}

\begin{proof}
\emph{Upper bound.}  Since $|e_i|\le 1$,
$|I|=|e_4 - e_1 e_2 e_3 |\le |e_4|+ |e_1 e_2 e_3|\le 2$.

\emph{Necessity of $|e_i|=1$.}
Saturation $|I|=2$ requires $|e_1 e_2 e_3|=1$ \emph{and} $|e_4|=1$.
In particular each $|e_i|=1$.  For a Hermitian operator $O$ with
eigenvalues $\pm1$ and a normalised state $\rho$, $|\langle O\rangle_\rho|=1$
forces $\rho$ to be supported in a single eigenspace of $O$.  Hence
$\rho$ lies in the simultaneous $\pm1$-eigenspace of all four $O_i$; in
particular every pure-state component $\ket{\psi}$ in the support
satisfies $O_i\ket\psi=\epsilon_i\ket\psi$.

\emph{Necessity of $\nv_1\perp\nv_2$.}
Acting~\eqref{eq:O1O2O3-O4} on a common eigenstate yields
\begin{equation}\label{eq:eigvalsum}
(\epsilon_1\epsilon_2\epsilon_3+\epsilon_4)\ket\psi
=2c\,\bigl(\sigma_{\nv_1}\otimes\sigma_{\nv_2}\otimes\sigma_{\nv_1}\bigr)\ket\psi.
\end{equation}
The operator on the right is unitary with spectrum $\{\pm1\}$, so the
right-hand side has norm exactly $2|c|$.  The left-hand side has norm
$|\epsilon_1\epsilon_2\epsilon_3+\epsilon_4|\in\{0,2\}$.  Two cases
arise:
\begin{itemize}
\item If $\epsilon_1\epsilon_2\epsilon_3=\epsilon_4$ then
$I=\epsilon_1\epsilon_2\epsilon_3-\epsilon_4=0$, contradicting
$|I|=2$.
\item If $\epsilon_1\epsilon_2\epsilon_3=-\epsilon_4$ then the left-hand
side of~\eqref{eq:eigvalsum} vanishes, forcing $c=\nv_1\cdot\nv_2=0$.
\end{itemize}
Hence saturation requires $\nv_1\perp\nv_2$ and the sign relation
$\epsilon_1\epsilon_2\epsilon_3=-\epsilon_4$, equivalently
$\epsilon_1\epsilon_2\epsilon_3\epsilon_4=-1$, which is consistent
with~\eqref{eq:GHZ-stabiliser}.

\emph{Achievability and LU-equivalence to GHZ.}
For $\nv_1\perp\nv_2$ the four operators in~\eqref{eq:OOOO}
commute and generate an abelian group of order eight, with
relation~\eqref{eq:GHZ-stabiliser}.  Three of them (say $O_1,O_2,O_3$)
are independent and have a joint orthonormal basis of $\pm1$
eigenstates~\cite{gottesman1997}; the fourth eigenvalue is then
$\epsilon_4=-\epsilon_1\epsilon_2\epsilon_3$.  Hence
$I=\epsilon_4- \epsilon_1\epsilon_2\epsilon_3= - 2\epsilon_1\epsilon_2\epsilon_3
=\pm2$ on each of the eight common eigenstates.

When $\nv_1=\xh,\nv_2=\yh$ this basis is the textbook GHZ basis
\begin{equation}\label{eq:GHZbasis}
\tfrac{1}{\sqrt{2}}\bigl(\ket{ijk}\pm\ket{\bar i\bar j\bar k}\bigr),
\qquad i,j,k\in\{0,1\},
\end{equation}
where $\bar 0=1$, $\bar1=0$.  All eight states in~\eqref{eq:GHZbasis}
are LU-equivalent to $\ket{\GHZ}$ via local Pauli flips and phase
gates.  For arbitrary orthogonal $(\nv_1,\nv_2)$ the joint basis is the
image of~\eqref{eq:GHZbasis} under the local unitary $U^{\otimes3}$
that sends $(\xh,\yh)$ to $(\nv_1,\nv_2)$, and is therefore again
LU-equivalent to $\ket{\GHZ}$.
\end{proof}

\textit{Bell-type interpretation.}
Proposition~\ref{thm:max} is more than a Cirel'son-type bound~\cite{cirelson1980}:
the LHV identity~\eqref{eq:LHV-zero} predicts $I=0$ whereas quantum
mechanics achieves $|I|=2$ exactly on the GHZ orbit.  In other words,
$I(\nv_1,\nv_2)$ is an \emph{equality-form} GHZ witness whose violation
\emph{measures} how close $\rho$ is, in the sense of stabiliser
correlations, to a GHZ-type state.  We emphasise that the saturation
condition is sharper than what is available for the linear Mermin
operator: the closed-form upper bound of Siddiqui and
Sazim~\cite{siddiqui2019} for the linear functional
$M_3=\langle\sigma_{xxx}\rangle-\langle\sigma_{xyy}\rangle-\langle\sigma_{yxy}\rangle-\langle\sigma_{yyx}\rangle$
involves diagonalisation of state-dependent correlation matrices,
whereas our multiplicative bound \emph{factorises} and is achieved on
the pure GHZ orbit alone.  The complementary ``inequality-free''
viewpoint of Cabello~\cite{cabello2002}, which produces an
all-versus-nothing separator of GHZ and W via incompatible composite
observables, parallels but does not subsume the present quantitative
witness.

Because $I$ is built from the deterministic local-hidden-variable
identity~\eqref{eq:LHV-identity}, one might read $|I|>0$ as evidence of
nonlocality.  This reading is not warranted: eq.~\eqref{eq:LHV-zero} is
an algebraic identity holding run by run for a \emph{deterministic}
value assignment, not a Bell inequality, so $|I|>0$ is not by itself a
violation of a local-realistic bound for general stochastic LHV models
with shared randomness, and $I=0$ does not certify an LHV model.  The
distinction matters because for three parties entanglement and
nonlocality are not interchangeable: by Gisin's theorem every entangled
pure two-qubit state violates a Bell inequality~\cite{gisin1991}, but no
tripartite counterpart is known, and Zhang \emph{et
al.}~\cite{zhang2026} established the equivalence of genuine
multipartite nonlocality, steering and entanglement only under an
additional multicopy/network-inflation assumption, leaving the
single-copy question open.  Throughout, $I$ and $\Elu$ quantify
GHZ-type stabiliser correlation (entanglement structure) and not
nonlocality; accordingly $\Elu(\rho)>1/2$ is (conditional, see (P6)
below) evidence of genuine tripartite \emph{entanglement}, and we
describe $\Elu$ as \emph{local-frame independent} rather than device
independent: the parties need share no reference frame, but the
quantity is not certified in the Bell sense.
\section{Evaluation on canonical pure states}\label{sec:eval}

We now evaluate $I$ on prototypical three-qubit pure states.

\subsection{The Ac\'in canonical family}\label{sec:acin}

Up to local unitaries, every pure three-qubit state can be brought to
the Ac\'in canonical form~\cite{acin2000,acin2001}
\begin{equation}\label{eq:acin}
\ket{\Psi}=\lambda_0\ket{000}+\lambda_1\mathrm{e}^{\mathrm{i}\varphi}\ket{100}
+\lambda_2\ket{101}+\lambda_3\ket{110}+\lambda_4\ket{111},
\end{equation}
with $\lambda_i\ge 0$, $\sum_i\lambda_i^2=1$, and
$\varphi\in[0,\pi]$.  This parametrisation has five independent real parameters.

A direct calculation (App.~\ref{app:acin}) yields
\begin{align}
\langle\sigma_x^{\otimes3}\rangle_{\ket\Psi}&=2\lambda_0\lambda_4,\\
\langle\sigma_x\sigma_y\sigma_y\rangle_{\ket\Psi}
=\langle\sigma_y\sigma_x\sigma_y\rangle_{\ket\Psi}
=\langle\sigma_y\sigma_y\sigma_x\rangle_{\ket\Psi}
&=-2\lambda_0\lambda_4,
\end{align}
\emph{independently of} $\lambda_1,\lambda_2,\lambda_3,\varphi$.
Substituting in~\eqref{eq:Idef} we obtain the closed form
\begin{equation}\label{eq:I-acin}
\;I(\xh,\yh;\ket\Psi)= 2\lambda_0\lambda_4\bigl[4(\lambda_0\lambda_4)^2+1\bigr].\;
\end{equation}
Maximising $|I|$ over the simplex $\{\lambda_i^2:\sum_i\lambda_i^2=1,\,
\lambda_i\ge0\}$ reduces to maximising the single product
$\mu\equiv\lambda_0\lambda_4$, which by Cauchy--Schwarz is bounded by
$\mu\le 1/2$ with equality iff $\lambda_0=\lambda_4=1/\sqrt{2}$ and
all other $\lambda_i=0$.  This is the GHZ state, and
\begin{equation}
|I_{\GHZ}(\xh,\yh)|=8\,(\tfrac12)^3+2\cdot\tfrac12=1+1=2,
\end{equation}
in agreement with Proposition~\ref{thm:max}.  The function
$|I(\xh,\yh)|=8\mu^3+2\mu$ is monotonically increasing in $\mu$ on
$[0,1/2]$ [Fig.~\ref{fig:main}(a)].  In particular every state with
$\lambda_0\lambda_4=0$, which includes all states whose Ac\'in form has no
$\ket{111}$ component, gives $I(\xh,\yh)=0$.  This is the case for
the W state, see Sec.~\ref{sec:W} below.

A further consequence of~\eqref{eq:I-acin}: on the GHZ subfamily
$\lambda_1=\lambda_2=\lambda_3=0$, $\varphi=0$, namely
$\ket{\Psi(\beta)}=\cos\beta\ket{000}+\sin\beta\ket{111}$, one has
$\lambda_0\lambda_4=\tfrac12\sin(2\beta)$ and hence
\begin{equation}
I(\xh,\yh;\ket{\Psi(\beta)})= \sin^3(2\beta) + \sin(2\beta),
\end{equation}
which is determined solely by the Schmidt angle and reaches the GHZ
extremum $|I|=2$ at $\beta=\pi/4$.

\textit{Relation to the three-tangle.}
On the Ac\'in canonical family the Coffman--Kundu--Wootters
three-tangle equals $\tau_3=4\lambda_0^2\lambda_4^2$~\cite{ckw2000}.
Equation~\eqref{eq:I-acin} therefore admits the equivalent form
\begin{equation}\label{eq:I-tau3}
I(\xh,\yh;\ket\Psi)= \sqrt{\tau_3}\,(\tau_3+1),
\end{equation}
exhibiting $|I(\xh,\yh)|$ as a strictly increasing degree-3 polynomial
in $\sqrt{\tau_3}$ on the canonical pure-state slice.  However, this
identification is \emph{specific to the canonical form and to the
direction pair $(\xh,\yh)$}: on a generic state, neither
representative of an LU orbit need be in Ac\'in form, and $I$ depends
on $(\nv_1,\nv_2)$.  Accordingly, the invariant measure $\Elu$ of
Sec.~\ref{sec:measure} is
\emph{not} a function of $\tau_3$ on arbitrary states: most strikingly,
$\tau_3(\W)=0$~\cite{ckw2000} while $\Elu(\W)=35/54$ (see
Sec.~\ref{sec:measure} below).  This separation reflects that $\Elu$ couples
to the entire \emph{stabiliser-correlation} structure of the GHZ
class, not merely to the residual hyperdeterminant captured by
$\tau_3$~\cite{eltschka2008}.

\begin{figure*}[t]
\centering
\includegraphics[width=0.85\textwidth]{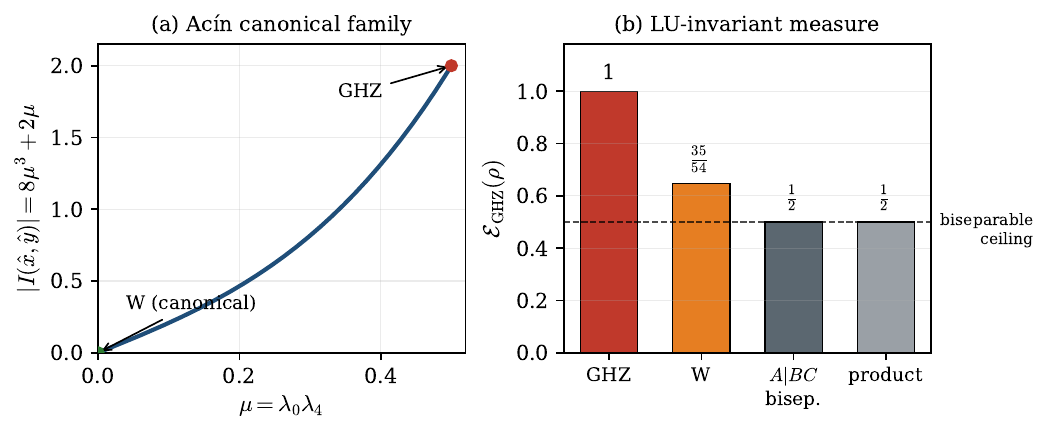}
\caption{(a) The functional $|I(\xh,\yh)|$ evaluated on the Ac\'in
canonical family~\eqref{eq:acin} depends only on the product
$\mu=\lambda_0\lambda_4$, with the closed form $|I|=8\mu^3+2\mu$
[Eq.~\eqref{eq:I-acin}].  The maximum value $2$ is attained at
$\mu=1/2$, the GHZ state.  The W state has $\mu=0$ in its canonical
form, giving $I(\xh,\yh)=0$.  (b) The local-unitary invariant measure
$\Elu$ of Sec.~\ref{sec:measure}, obtained by maximising the
correlation structure over \emph{independent} orthonormal frames on the
three parties, for four representative classes of three-qubit states.
The GHZ state attains the maximal value $\Elu=1$; the W state achieves
$35/54\approx 0.648$ (the independent-frame optimisation does not exceed
the aligned-frame value here, due to the permutation symmetry of
$\ket\W$); biseparable $A|BC$ states and fully separable states are both
bounded by, and attain, $1/2$ (a fixed-cut bound; convex mixtures across
any two cuts still obey it, while the three-cut mixture of
Sec.~\ref{sec:counterexample} exceeds it).  Values are obtained by
multistart Powell optimisation over $120$ random initial frame triples
and verified to machine precision against the exact values of
Sec.~\ref{sec:measure}; full numerical details are given in
App.~\ref{app:convex}.}
\label{fig:main}
\end{figure*}

\subsection{The W state}\label{sec:W}

The W state $\ket\W=(\ket{001}+\ket{010}+\ket{100})/\sqrt{3}$ is the
representative of the W class, which is SLOCC-inequivalent to the GHZ
class~\cite{dvc2000}.  Direct computation of all 27 three-qubit Pauli
correlators yields
\begin{equation}\label{eq:Wcorr}
\begin{aligned}
\langle\sigma_z^{\otimes3}\rangle_\W&=-1,\\
\langle\sigma_a \otimes \sigma_a \otimes \sigma_z\rangle_\W&=\langle\sigma_a \otimes \sigma_z \otimes \sigma_a\rangle_\W \\
&=\langle\sigma_z\otimes \sigma_a\otimes  \sigma_a\rangle_\W =\tfrac23,
\end{aligned}
\end{equation}
for $a\in\{x,y\}$,
all other components vanishing.

\textit{Special directions.}
With $\nv_1=\xh$, $\nv_2=\yh$, every correlator
appearing in $I$ vanishes by~\eqref{eq:Wcorr}; therefore
\begin{equation}\label{eq:I-W-xy}
I_\W(\xh,\yh)=0.
\end{equation}
By contrast, with $\nv_1=\zh$, $\nv_2=\xh$ (or any
$\nv_2$ in the $\xh\!\text{-}\yh$ plane, by the $\zh$-rotation symmetry of
$\ket\W$) one obtains
\begin{equation}\label{eq:I-W-zx}
I_\W(\zh,\xh)=(-1)-\bigl(\tfrac23\bigr)^3=-1-\frac{8}{27}=-\frac{35}{27}.
\end{equation}

\textit{Maximum over directions.}
For arbitrary orthogonal pairs we set $a_3\equiv\zh\cdot\nv_1$,
$b_3\equiv\zh\cdot\nv_2$.  Using~\eqref{eq:Wcorr} together with
$\nv_1\cdot\nv_2=0$ and $\nv_2\cdot\nv_2=1$ one obtains the explicit
formula (derived in App.~\ref{app:Wmax})
\begin{equation}\label{eq:I-W-formula}
I_\W(\nv_1,\nv_2)=
 a_3\bigl(2-3a_3^2\bigr) -a_3^3\Bigl(\tfrac23-3b_3^2\Bigr)^{\!3}.
\end{equation}
A short Lagrangian analysis (App.~\ref{app:Wmax}) shows that the global
extremum of $|I_\W|$ on the constraint surface
$a_3^2+b_3^2\le 1$ is attained at $(a_3,b_3)=(\pm1,0)$, giving
\begin{equation}\label{eq:I-W-max}
\max_{\nv_1,\nv_2}|I_\W(\nv_1,\nv_2)|=\frac{35}{27}\approx 1.296.
\end{equation}
This is strictly less than the GHZ value $2$, reflecting the well-known
fact that the W state lies outside the GHZ class~\cite{dvc2000}.
Equation~\eqref{eq:I-W-max} can therefore be regarded as a sharp
\emph{quantitative} separator of the two SLOCC classes via $I$.  To
the best of our knowledge the value $35/27$ has not previously appeared
in the Mermin-functional literature: under the linear Mermin operator,
recent analytical work~\cite{siddiqui2019} reports
$\max\langle M_3\rangle_\W\approx 3.046$ (out of an algebraic ceiling
of $4$), while the present multiplicative functional gives the simpler
closed form~\eqref{eq:I-W-max}.

\textit{A remark on orthogonal versus general settings.}
The linear and multiplicative functionals differ structurally.  The Mermin maximum
$\max\langle M_3\rangle_\W\approx 3.046$ is achieved only over
\emph{general} (non-orthogonal) measurement directions; if one restricts
the two Mermin settings to be mutually orthogonal on each party, the W
state attains exactly $\langle M_3\rangle_\W=3$, and the extra
$0.046$ requires tilting the settings away from orthogonality.  We have
verified both numbers numerically (the orthogonal value $3$ analytically,
the general value $3.046$ by multistart optimisation over unconstrained
unit vectors).  Our functional $I$ behaves oppositely: by
Proposition~\ref{thm:max} its algebraic ceiling is reached \emph{exactly}
at orthogonal settings $\nv_1\perp\nv_2$, and tilting away from
orthogonality can only decrease $|I|$ on the GHZ orbit.  Thus the GHZ
paradox is most naturally a statement about orthogonal (maximally
complementary) local bases, which is the regime singled out by
the multiplicative form.

\subsection{Biseparable and product states}\label{sec:bisep}

For an $A|BC$ biseparable state $\rho=\rho_A\otimes\rho_{BC}$,
all four expectations factorise as
$\eopV{a}{b}{c}=\langle\sigma_{\nv_a}\rangle_A\cdot
\langle\sigma_{\nv_b}\otimes\sigma_{\nv_c}\rangle_{BC}$.  When
$\rho_A=\ket{0}\bra{0}$ and $\rho_{BC}=\ket{\Phi^+}\bra{\Phi^+}$ with
$\ket{\Phi^+}=(\ket{00}+\ket{11})/\sqrt2$, taking $\nv_1=\zh,\nv_2=\xh$
yields $e_1=\langle\sigma_z\rangle\langle\sigma_{xx}\rangle=1\cdot 1=1$,
$e_2=e_3=0$, $e_4=\langle\sigma_z\rangle\langle\sigma_{zz}\rangle
=1\cdot 1=1$, so $I=1$.  Numerical optimisation over
all directions gives
$\sup_{\nv_1\perp\nv_2}|I_{\textnormal{bisep}}|=1$ for the shared-frame
functional, in agreement with Proposition~\ref{thm:max} forbidding
saturation by biseparable states.  Fully separable (product) states
likewise reach at most $1$.  The corresponding values of the LU-invariant
measure $\Elu$ of Sec.~\ref{sec:measure}, which optimises over independent
local frames and carries an extra factor $1/2$, are
$\Elu(\textnormal{bisep})=\Elu(\textnormal{prod})=1/2$
[Fig.~\ref{fig:main}(b)]; see Sec.~\ref{sec:measure} and
App.~\ref{app:convex} for the derivation.

\textit{Random pure states.}
Sampling Haar-random three-qubit pure states and numerically
maximising $|I|$ over directions, we find $\sup|I|<2$ in every case;
random pure states cluster around values $0.4$--$1.6$, consistent with
the fact that the GHZ orbit is a measure-zero subset of pure-state
space.

\section{A genuine GHZ-type entanglement measure}\label{sec:measure}

The quantitative content of Proposition~\ref{thm:max} suggests building a
state functional by optimising $|I|$ over measurement directions.  Doing
this naively, however, runs into an invariance subtlety that we now make
explicit and then resolve.

\subsection{The aligned-frame quantity and its invariance defect}
\label{sec:aligned}

The most direct construction reuses the \emph{same} ordered orthonormal
pair $(\nv_1,\nv_2)$ on all three parties,
\begin{equation}\label{eq:Ealigned}
\Eghz(\rho)\;=\;\frac{1}{2}\sup_{\nv_1\perp\nv_2}
\bigl|I(\nv_1,\nv_2;\rho)\bigr| ,
\end{equation}
which we call the \emph{aligned-frame} quantity.  By
Proposition~\ref{thm:max}, $\Eghz\in[0,1]$ and $\Eghz(\rho)=1$ if and only
if $\rho$ is LU equivalent to $\ket{\GHZ}$.  It is, however, \emph{not}
invariant under general local unitaries.  A collective rotation
$U^{\otimes3}$ with $U\in SU(2)$ does leave~\eqref{eq:Ealigned} unchanged,
because the induced $R\in SO(3)$ maps the orthonormal pair
$(\nv_1,\nv_2)\mapsto(R^{-1}\nv_1,R^{-1}\nv_2)$, again orthonormal, and the
supremum runs over all such pairs.  Thus $\Eghz$ is invariant under the
diagonal subgroup $\{U\otimes U\otimes U\}$.  Under \emph{independent}
local unitaries $U_A\otimes U_B\otimes U_C$ the three parties' Pauli bases
rotate by three \emph{different} elements of $SO(3)$, which a single
shared frame cannot track.  Numerically the violation is dramatic: for
the GHZ state one finds $\Eghz(\ket{\GHZ})=1$, but under randomly drawn
independent local unitaries the same state yields values that scatter
between roughly $0.52$ and $0.67$; the W state's aligned-frame value
$0.648$ likewise drops to $0.36$--$0.49$.

\subsection{Independent local frames and the invariant measure}
\label{sec:indep}

The fix is to allow each party its own orthonormal frame.  Let
$(\hat a_1,\hat a_2)$, $(\hat b_1,\hat b_2)$, $(\hat c_1,\hat c_2)$ be
orthonormal pairs on $A$, $B$, $C$, and set
\begin{multline}\label{eq:Istar}
\Igen(\hat a_1,\hat a_2,\hat b_1,\hat b_2,\hat c_1,\hat c_2;\rho)
=\langle\sigma_{\hat a_1}\sigma_{\hat b_1}\sigma_{\hat c_1}\rangle\\
-\langle\sigma_{\hat a_1}\sigma_{\hat b_2}\sigma_{\hat c_2}\rangle
 \langle\sigma_{\hat a_2}\sigma_{\hat b_1}\sigma_{\hat c_2}\rangle
 \langle\sigma_{\hat a_2}\sigma_{\hat b_2}\sigma_{\hat c_1}\rangle,
\end{multline}
the natural independent-frame analogue of~\eqref{eq:Idef} (here
$\sigma_{\hat a}\sigma_{\hat b}\sigma_{\hat c}$ abbreviates
$\sigma_{\hat a}\otimes\sigma_{\hat b}\otimes\sigma_{\hat c}$).  It is
governed by the four Hermitian observables
\begin{equation}\label{eq:Aops}
\begin{aligned}
A_1&=\sigma_{\hat a_1}\otimes\sigma_{\hat b_1}\otimes\sigma_{\hat c_1},&
A_2&=\sigma_{\hat a_1}\otimes\sigma_{\hat b_2}\otimes\sigma_{\hat c_2},\\
A_3&=\sigma_{\hat a_2}\otimes\sigma_{\hat b_1}\otimes\sigma_{\hat c_2},&
A_4&=\sigma_{\hat a_2}\otimes\sigma_{\hat b_2}\otimes\sigma_{\hat c_1},
\end{aligned}
\end{equation}
so that $\Igen=\langle A_1\rangle-\langle A_2\rangle\langle A_3\rangle
\langle A_4\rangle$.  Because each party contributes either a repeated
vector or its orthogonal partner, any two of the $A_i$ differ by an
\emph{anticommuting} single-qubit factor on an even number of sites; one
checks (App.~\ref{app:convex}) that
\begin{equation}\label{eq:Acommute}
[A_i,A_j]=0\quad(\text{all }i,j),\qquad A_1A_2A_3A_4=-\openone .
\end{equation}
Define the \emph{local-unitary invariant} measure
\begin{equation}\label{eq:Elu}
\Elu(\rho)\;=\;\frac{1}{2}\,
\sup_{\substack{\hat a_1\perp\hat a_2,\ \hat b_1\perp\hat b_2,\\
\hat c_1\perp\hat c_2}}
\bigl|\Igen(\hat a_1,\hat a_2,\hat b_1,\hat b_2,\hat c_1,\hat c_2;\rho)\bigr|.
\end{equation}

\begin{proposition}[Identification and LU invariance]\label{prop:luequiv}
For every three-qubit state $\rho$, writing
$V=U_A\otimes U_B\otimes U_C$,
\begin{equation}\label{eq:maxU}
\Elu(\rho)=\max_{U_A,U_B,U_C\in SU(2)}
\Eghz\!\left(V\rho V^\dagger\right),
\end{equation}
and consequently $\Elu$ is invariant under arbitrary local unitaries.
\end{proposition}

\begin{proof}
Conjugating $\rho$ by $U_A\otimes U_B\otimes U_C$ and then evaluating the
aligned functional $I(\nv_1,\nv_2)$ replaces each single-qubit observable
$\sigma_{\nv}$ on party $X$ by $\sigma_{R_X\nv}$, where $R_X\in SO(3)$ is
the rotation associated with $U_X$.  Thus the aligned pair $(\nv_1,\nv_2)$
is seen by the three parties as $(R_A\nv_1,R_A\nv_2)$,
$(R_B\nv_1,R_B\nv_2)$, $(R_C\nv_1,R_C\nv_2)$.  Since $SO(3)$ acts
transitively on ordered orthonormal pairs, as $(U_A,U_B,U_C,\nv_1,\nv_2)$
range over all values these triples of pairs range over \emph{all} triples
of independent orthonormal frames in~\eqref{eq:Elu}.  Hence the right-hand
side of~\eqref{eq:maxU} equals $\Elu(\rho)$.  Invariance is immediate:
the maximisation over all local unitaries absorbs any further local
unitary applied to $\rho$.
\end{proof}

We can now reinstate, correctly, the saturation theorem at the level of
the invariant measure.

\begin{proposition}[Genuine GHZ indicator]\label{prop:satlu}
$\Elu(\rho)\in[0,1]$ for every three-qubit state, and $\Elu(\rho)=1$ if
and only if $\rho$ is LU equivalent to $\ket{\GHZ}$.
\end{proposition}

\begin{proof}
Each $A_i$ is a $\pm1$-valued Hermitian unitary, so
$|\langle A_i\rangle|\le1$ and $|\Igen|\le|\langle A_1\rangle|
+|\langle A_2\rangle\langle A_3\rangle\langle A_4\rangle|\le2$, giving
$\Elu\le1$.  Saturation $|\Igen|=2$ forces $|\langle A_i\rangle|=1$ for
all $i$, hence $\rho$ is supported in a joint $\pm1$ eigenspace of the
four commuting operators~\eqref{eq:Acommute}.  These generate an abelian
group with three independent generators (App.~\ref{app:convex}), whose
eight common eigenspaces are one-dimensional and span $\mathcal H$; on
each, $A_1A_2A_3A_4=-\openone$ gives
$\epsilon_1\epsilon_2\epsilon_3\epsilon_4=-1$ and
$\Igen=\epsilon_1-\epsilon_2\epsilon_3\epsilon_4=2\epsilon_1=\pm2$.  Thus
$\rho$ is the pure projector onto one such common eigenstate.  Applying
independent local rotations that send $(\hat a_1,\hat a_2)$,
$(\hat b_1,\hat b_2)$, $(\hat c_1,\hat c_2)$ each to $(\xh,\yh)$ turns
$\{A_1,A_2,A_3,A_4\}$ into the standard GHZ--Mermin stabiliser operators
$\{\sigma_{xxx},\sigma_{xyy},\sigma_{yxy},\sigma_{yyx}\}$, whose common
eigenbasis is the GHZ basis~\eqref{eq:GHZbasis}.  Hence the eigenstate is
LU equivalent to $\ket{\GHZ}$, and conversely $\ket{\GHZ}$ attains
$\Elu=1$ at orthogonal frames.
\end{proof}

\subsection{Values on canonical states}\label{sec:lu-values}

\textit{GHZ and W.}
Proposition~\ref{prop:satlu} gives $\Elu(\ket{\GHZ})=1$.  For the W state
the maximisation over independent frames does \emph{not} exceed the
aligned-frame value, a consequence of the permutation symmetry of
$\ket\W$ and its invariance under collective rotations about $\zh$: the
optimum is already realised by the aligned configuration $\nv_1=\zh$,
$\nv_2\in\xh\text{-}\yh$ of Sec.~\ref{sec:W}, so that
\begin{equation}\label{eq:Elu-W}
\Elu(\ket\W)=\tfrac12\cdot\tfrac{35}{27}=\frac{35}{54}\approx0.648,
\end{equation}
which we have confirmed to machine precision by independent-frame
optimisation.  The strict ordering
$\Elu(\GHZ)=1>\Elu(\W)=35/54$ persists.

\textit{Product states (exact).}
For a pure product state the correlation tensor factorises,
$\langle\sigma_i\sigma_j\sigma_k\rangle=a_i b_j c_k$ with unit Bloch
vectors $\vec a,\vec b,\vec c$.  Writing $p_r=\vec a\!\cdot\!\hat a_r$,
$q_r=\vec b\!\cdot\!\hat b_r$, $s_r=\vec c\!\cdot\!\hat c_r$
($r=1,2$, $p_1^2+p_2^2\le1$, etc.),
\begin{equation}
\Igen=p_1q_1s_1\bigl(1-p_2^2q_2^2s_2^2\bigr),
\end{equation}
whence $|\Igen|\le|p_1q_1s_1|\le1$, with equality at
$\hat a_1=\vec a$, $\hat b_1=\vec b$, $\hat c_1=\vec c$.  Therefore
\begin{equation}\label{eq:Elu-prod}
\Elu(\textnormal{product})=\tfrac12 .
\end{equation}

\textit{Biseparable states.}
For $A|BC$ biseparable $\rho=\rho_A\otimes\rho_{BC}$ let $u=\vec
a\!\cdot\!\hat a_1$, $u'=\vec a\!\cdot\!\hat a_2$ (with $u^2+u'^2\le1$) and
let $f_{jk}=\hat b_j^{\mathsf T}S\,\hat c_k$ be the relevant correlators of
the two-qubit block, where $S_{jk}=\langle\sigma_j\otimes\sigma_k
\rangle_{BC}$.  Since every singular value of a two-qubit correlation
matrix is at most one, the $2\times2$ array $[f_{jk}]=B^{\mathsf T}SC$
(with $B=[\hat b_1,\hat b_2]$, $C=[\hat c_1,\hat c_2]$) has operator norm
$\le1$, so $|f_{jk}|\le1$ and $f_{11}^2+f_{12}^2\le1$ along each row and
column.  Then
\begin{equation}
|\Igen|=|u\,f_{11}-u\,u'^2 f_{22}f_{12}f_{21}|
\le|u|\,(1+u'^2)\le \frac{4\sqrt6}{9},
\end{equation}
the last step maximising $|u|(1+u'^2)$ over $u^2+u'^2\le1$.  Hence
\begin{equation}\label{eq:Elu-bisep-bound}
\Elu(A|BC)\;\le\;\frac{2\sqrt6}{9}\approx0.544\;<\;1 ,
\end{equation}
already enough to exclude biseparable states from the maximal value.

The triangle bound~\eqref{eq:Elu-bisep-bound} is, however, not tight.
The sharp value is
\begin{equation}\label{eq:Elu-bisep}
\sup_{\rho\,\in\,A|BC}\Elu(\rho)=\tfrac12 ,
\end{equation}
attained e.g.\ by $\ket0_A\otimes\ket{\Phi^+}_{BC}$ and coinciding with
the product-state value~\eqref{eq:Elu-prod}.  Zhang and Ma have proved
that this value is exact for arbitrary density factors, neither of which
need be pure~\cite{ZhangMa2026}; a Lean formalisation accompanies their
proof.  The same value
holds for the $B|AC$ and $C|AB$ cuts by symmetry.

Equation~\eqref{eq:Elu-bisep} is a statement about a \emph{fixed} cut.
The bound does \emph{not} extend to the convex hull of all three families
(Sec.~\ref{sec:counterexample}); it does, however, extend to convex
mixtures across any two cuts~\cite{ZhangMa2026}, as we show next.

\subsection{Mixtures across two bipartitions}\label{sec:twocut}

Let $\mathcal{S}_{AB}=\mathrm{conv}\bigl[\mathrm{Sep}(A|BC)\cup
\mathrm{Sep}(B|AC)\bigr]$, where $\mathrm{Sep}(A|BC)$ denotes the states
separable across the cut $A|BC$ (the factors may be mixed).  Zhang and Ma
have proved the following sharp bound for this convex
hull~\cite{ZhangMa2026}; it uses no convexity assumption on $\Elu$
itself.

For every $\rho\in\mathcal{S}_{AB}$ and every choice of local frames,
\begin{equation}\label{eq:paired}
|E_{111}|+|E_{221}|\le1,\qquad
|E_{122}|+|E_{212}|\le1 ,
\end{equation}
and consequently $|\Igen|\le(1+3|E_{111}|)/4\le1$, so that
$\Elu(\rho)\le1/2$~\cite{ZhangMa2026}.  The constant $1/2$ is sharp at
every mixing weight, and the same statements hold for any other pair of
cuts.

The bound is attained, for every $t\in[0,1]$, by
$\rho_t=t\ket{\psi_A}\!\bra{\psi_A}+(1-t)\ket{\psi_B}\!\bra{\psi_B}$ with
$\ket{\psi_A}=\ket0_A\ket{\Phi^+}_{BC}$ and
$\ket{\psi_B}=\ket0_B\ket{\Phi^+}_{AC}$, at the common frame
$(\hat z,\hat x)$; the four correlations are then $(1,t,1-t,0)$ and
$\Igen=1$~\cite{ZhangMa2026}.  The same reference also characterises the
equality case: for $\rho\in\mathcal{S}_{AB}$ one has $\Elu(\rho)=1/2$ if
and only if, at some frames, $\rho$ is supported in a single eigenspace
of a product observable
$\sigma_{\hat a}\otimes\sigma_{\hat b}\otimes\sigma_{\hat c}$.

Together with Sec.~\ref{sec:counterexample} this delimits the threshold
sharply: the bound $1/2$ holds for every fixed cut and for every two-cut
hull, and first fails when all three cuts are mixed.

\subsection{Summary of properties}\label{sec:lu-props}

\begin{proposition}\label{prop:measure}
The functional $\Elu$ satisfies:
\begin{enumerate}
\item[\textnormal{(P1)}] $0\le\Elu(\rho)\le 1$;
\item[\textnormal{(P2)}] $\Elu$ is invariant under arbitrary local
  unitaries \emph{(Proposition~\ref{prop:luequiv})};
\item[\textnormal{(P3)}] $\Elu(\rho)=1$ if and only if $\rho$ is LU
  equivalent to $\ket{\GHZ}\bra{\GHZ}$ \emph{(Proposition~\ref{prop:satlu})};
\item[\textnormal{(P4)}] $\Elu(\ket\W)=35/54\approx 0.648$;
\item[\textnormal{(P5)}] $\Elu(\textnormal{product})=1/2$ exactly;
  $\Elu\le1/2$ for every state separable across a fixed bipartition
  (Sec.~\ref{sec:lu-values}) and for every convex mixture across any
  \emph{two} bipartitions (Sec.~\ref{sec:twocut})~\cite{ZhangMa2026}, the
  value $1/2$ being
  attainable in both cases;
\item[\textnormal{(P6)}] the threshold $1/2$ does \emph{not} extend to
  mixtures across all three cuts, and $\Elu$ is \emph{not} convex: the
  three-cut mixture~\eqref{eq:rhomix} satisfies
  $\Elu=\tfrac{14}{27}>\tfrac12
   =\tfrac13[\Elu(\rho_A)+\Elu(\rho_B)+\Elu(\rho_C)]$
  (Sec.~\ref{sec:counterexample}).
\end{enumerate}
\end{proposition}

The failure of the convexity in property (P6) is specifically a
\emph{three-cut} phenomenon.  By the two-cut bound of Ref.~\cite{ZhangMa2026} the value
$1/2$ survives on every two-cut hull, so mixing alone is not the
obstruction; what matters is that all three bipartitions are involved.
For \emph{fixed} frames the integrand
$|\langle A_1\rangle-\langle A_2\rangle\langle A_3\rangle\langle A_4\rangle|$
is the modulus of a linear term minus a \emph{cubic} term in $\rho$ and
is not convex (midpoint violations in roughly $16\%$ of random pairs),
and the optimised measure inherits this failure through the explicit
three-component mixture of Sec.~\ref{sec:counterexample}.  In that
example mixing \emph{enhances} the measure: $\Elu(\rho)=14/27$ while each
component has $\Elu=1/2$, because the cubic term is non-zero only for the
mixture, where its three factors are supplied by three different
components.  The GHZ--W segment, by contrast, remains below its chord:
for instance
$\Elu\bigl(\tfrac12\ket{\GHZ}\!\bra{\GHZ}+\tfrac12\ket\W\!\bra\W\bigr)
\approx0.374$ against a chord value $\approx0.824$.  Convexity can
therefore fail without being visible along that particular line.

\textit{No GME threshold at $1/2$.}
A genuine-multipartite-entanglement witness above the biseparable
ceiling would require $\Elu\le1/2$ on the whole convex hull of
biseparable states.  The state~\eqref{eq:rhomix} is biseparable by
construction, yet $\Elu=14/27>1/2$; hence the implication
``$\Elu(\rho)>1/2$ $\Rightarrow$ $\rho$ is GME'' fails, and the
convex-roof quantity $G$ of Sec.~\ref{sec:locc} is not a GME witness
either.  The threshold does survive on the smaller classes: it holds for
a fixed cut (Eq.~\eqref{eq:Elu-bisep}) and, by
Ref.~\cite{ZhangMa2026}, for every convex hull of two cuts.

\textit{Comparison with other measures.}
Unlike the three-tangle $\tau_3$~\cite{ckw2000} or the
GHZ$\oplus$W-restricted closed forms of Eltschka \emph{et
al.}~\cite{eltschka2008}, both of which vanish on the W class,
$\Elu$ takes nontrivial values on the W class and on states that are
separable across a fixed cut: it tells GHZ from W (with
$\Elu(\GHZ)=1>\Elu(\W)=35/54$) and from such states
(with $\Elu(\textnormal{bisep})\le 1/2<\Elu(\W)$); the bound also holds
on any two-cut hull~\cite{ZhangMa2026} but
fails for the
three-cut mixture of Sec.~\ref{sec:counterexample}.  In contrast
to GHZ-fidelity-based witnesses
$F_\GHZ(\rho)=\bra{\GHZ}\rho\ket{\GHZ}$~\cite{toth2005,bourennane2004},
$\Elu$ is constructed from local Pauli measurements alone and, by
Proposition~\ref{prop:luequiv}, is invariant under arbitrary local
unitaries, so it does not depend on a choice of GHZ representative.
Compared to the genuine multipartite negativity~\cite{hofmann2014} or the
GME-concurrence of Ma \emph{et al.}~\cite{ma2011}, both of which involve
semidefinite-programming or convex-roof constructions, the aligned
functional $\Eghz$ has a closed form on every pure Ac\'in state; the
invariant $\Elu$ trades that closed form for an optimisation over three
local frames, while remaining local-frame
independent in the operational sense that only spin directions, agreed
locally, need be measured.

\subsection{Relation to GHZ fidelity witnesses}\label{sec:fidelity}
The most widely used detector of GHZ-type genuine multipartite
entanglement is the fidelity witness
$W_{\GHZ}=\tfrac12\openone-\ket{\GHZ}\!\bra{\GHZ}$, whose expectation is
negative, certifying genuine tripartite entanglement, precisely when
$F_{\GHZ}(\rho)=\bra{\GHZ}\rho\ket{\GHZ}>1/2$
\cite{toth2005,bourennane2004,guhne2009}.  Because the projector onto
$\ket\GHZ$ is the uniform average of the eight elements of the GHZ
stabiliser group,
\begin{multline}\label{eq:Fstab}
F_{\GHZ}(\rho)=\tfrac18\Bigl[1
+\langle\sigma_{zz\openone}\rangle+\langle\sigma_{\openone zz}\rangle
+\langle\sigma_{z\openone z}\rangle\\
+\langle\sigma_{xxx}\rangle-\langle\sigma_{xyy}\rangle
-\langle\sigma_{yxy}\rangle-\langle\sigma_{yyx}\rangle\Bigr],
\end{multline}
the fidelity is assembled from exactly the same laboratory data as
$I(\xh,\yh)$: the four Mermin correlators together with the three
two-body $zz$ correlators.  The two quantities nevertheless differ in
three structural respects, with different physical readings.
\emph{(i) Linear versus multiplicative.}  $F_{\GHZ}$ is linear in
$\rho$, hence automatically convex, which is what makes the threshold
$1/2$ a rigorous witness.  Our $I$ is a linear term minus a
\emph{cubic} one, and $\Elu$ is \emph{not} convex
(Sec.~\ref{sec:counterexample}); the threshold $1/2$ therefore does not
certify GME.  What the multiplicative form buys in
exchange is that the GHZ paradox enters as an
\emph{equality}~\eqref{eq:LHV-identity}, which yields the sharp
saturation statement of Proposition~\ref{thm:max}: $|I|=2$ holds
\emph{if and only if} the state lies on the GHZ orbit, whereas
$F_{\GHZ}>1/2$ is a sufficient but not necessary criterion.
\emph{(ii) Two-body content.}  Equation~\eqref{eq:Fstab} mixes two-body
($zz$) and three-body correlators, whereas $I$ uses three-body
correlators only, so $\Elu$ is blind to bipartite correlations and can
separate genuine three-body correlation from residual two-body crosstalk
in a device.  \emph{(iii) Reference-frame dependence.}  $F_{\GHZ}$ is
defined relative to a chosen representative $\ket{\GHZ}$ and changes
under independent local unitaries; a frame-free version requires a
maximisation over $SU(2)^{\otimes3}$ of the same numerical cost
as~\eqref{eq:Elu}.  By Proposition~\ref{prop:luequiv}, $\Elu$ is LU
invariant by construction.

The resulting physical interpretations are complementary: $F_{\GHZ}$
measures proximity in Hilbert space to a nominated target state, while
$\Elu$ measures how strongly the state realises the GHZ algebraic
paradox, and the induced orderings differ accordingly.  Maximised over
local unitaries, the GHZ fidelity of the W state is $3/4$ (verified
numerically to $10^{-6}$; note that this maximised quantity is not
itself a witness, since the threshold $1/2$ applies to a fixed target
state), against $\Elu(\W)=35/54\approx0.648$.  The fidelity
witness is therefore the sharper tool when a rigorous yes/no certificate with a
proven threshold is wanted; $\Elu$ is the more informative one when the
task is to grade states on a representative-free, purely three-body
scale, requiring only four measurement settings rather than seven
correlators.
\subsection{Behaviour under local operations and classical
communication}\label{sec:locc}
We now make precise in what sense $\Elu$ is, and is not, compatible with
the LOCC paradigm.  By Proposition~\ref{prop:luequiv}, $\Elu$ is exactly
invariant under the reversible elements of LOCC (the local unitaries),
and under relabelling of the parties, since~\eqref{eq:Elu} optimises
over all three frames independently.

$\Elu$ is, however, not an entanglement monotone.  This is not merely
unproven: it fails, and for a structural reason.  A pure product state
has $\Elu=1/2$ by~\eqref{eq:Elu-prod}, whereas the maximally mixed state
$\rho=\openone/8$ has all correlators zero and hence $\Elu=0$.  Local
preparation, a legitimate LOCC operation, maps $\openone/8$ to
$\ket{000}\!\bra{000}$ and thereby \emph{raises} $\Elu$ from $0$ to
$1/2$.  The obstruction is that $\Elu$ does not vanish on separable
states: the value $1/2$ is a correlation offset carried by every product
state, not a residue of entanglement.  The natural repair is to
subtract the biseparable ceiling,
\begin{equation}\label{eq:Gdef}
G(\rho)\;=\;\max\{0,\;2\Elu(\rho)-1\}\;\in\;[0,1],
\end{equation}
which vanishes on every pure product and pure biseparable state (but
not, by Sec.~\ref{sec:counterexample}, on every mixed biseparable
state), equals
$1$ exactly on the GHZ orbit (Proposition~\ref{prop:satlu}), and takes
the value $G(\ket\W)=2\cdot\tfrac{35}{54}-1=\tfrac{8}{27}\approx0.296$;
being built from $\Elu$, $G$ is LU invariant.  Whether its convex roof
is non-increasing on average under LOCC remains an open question, left
for future work.

One monotonicity statement can be made exactly.  Let $\Lambda_q$ be the
single-qubit depolarising channel, applied independently to each of the
three qubits.  Every three-body correlator is then scaled by $q^3$, so
that
\begin{equation}\label{eq:depol-bound}
\Elu\bigl(\Lambda_q^{\otimes3}(\rho)\bigr)\;\le\;q^3
\qquad\text{for every }\rho .
\end{equation}
For the GHZ state the value is attained in closed form,
\begin{equation}\label{eq:depol-ghz}
\Elu\bigl(\Lambda_q^{\otimes3}(\ket\GHZ\!\bra\GHZ)\bigr)
=\tfrac12\bigl(q^3+q^9\bigr),
\end{equation}
confirmed numerically to $10^{-5}$ over $q\in\{0.6,\dots,1\}$.  Thus
$\Elu$ degrades smoothly and monotonically under local depolarising
noise, and~\eqref{eq:depol-ghz} can be inverted to extract an effective
local-noise parameter from the four measured correlators.

\subsection{A counterexample to convexity}\label{sec:counterexample}

Consider the permutation-symmetric state
\begin{equation}\label{eq:rhomix}
\rho=\frac13\Bigl(\ket0\!\bra0_A\otimes\ket{\Phi^+}\!\bra{\Phi^+}_{BC}
+\text{two permutations}\Bigr),
\end{equation}
with $\ket{\Phi^+}=(\ket{00}+\ket{11})/\sqrt2$,
i.e.\ the equal mixture of $\ket0_A\ket{\Phi^+}_{BC}$,
$\ket0_B\ket{\Phi^+}_{AC}$ and $\ket0_C\ket{\Phi^+}_{AB}$.  Each
component is biseparable (a product across one cut), so $\rho$ is
biseparable in the standard convex-hull sense.  All three cuts are
essential: by the two-cut bound~\cite{ZhangMa2026}, any convex mixture of states
separable across only \emph{two} cuts still obeys $\Elu\le1/2$.  Only
seven of the $27$ three-body correlators are non-zero:
\begin{equation}\label{eq:Tmix}
\begin{aligned}
&\langle\sigma_z\sigma_z\sigma_z\rangle=1,\\
&\langle\sigma_z\sigma_x\sigma_x\rangle
=\langle\sigma_x\sigma_z\sigma_x\rangle
=\langle\sigma_x\sigma_x\sigma_z\rangle=\tfrac13,\\
&\langle\sigma_z\sigma_y\sigma_y\rangle
=\langle\sigma_y\sigma_z\sigma_y\rangle
=\langle\sigma_y\sigma_y\sigma_z\rangle=-\tfrac13 .
\end{aligned}
\end{equation}
(All other three-body correlators vanish; the non-zero two-body
correlators do not enter $\Igen$.)

Take the \emph{same} orthonormal pair at each party,
\begin{equation}\label{eq:optframe}
(\hat a_1,\hat a_2)=(\hat b_1,\hat b_2)=(\hat c_1,\hat c_2)
=(\zh,\yh).
\end{equation}
Then $A_1=\sigma_z\sigma_z\sigma_z$,
$A_2=\sigma_z\sigma_y\sigma_y$, $A_3=\sigma_y\sigma_z\sigma_y$,
$A_4=\sigma_y\sigma_y\sigma_z$, and eq.~\eqref{eq:Tmix} gives
\begin{equation}\label{eq:Amix}
\langle A_1\rangle=1,\qquad
\langle A_2\rangle=\langle A_3\rangle=\langle A_4\rangle=-\tfrac13,
\end{equation}
so that
\begin{equation}\label{eq:I-mix}
\Igen=1-\Bigl(-\tfrac13\Bigr)^{3}
=1+\tfrac1{27}=\frac{28}{27},
\end{equation}
hence $\Elu(\rho)\ge\frac{14}{27}\approx0.5185$.

Numerically this is the supremum: a multistart search over the nine
frame angles (App.~\ref{app:convex}) returns
$\max|\Igen|=28/27=1.0370\ldots$ to machine precision: $1200$ random
restarts of a compass search, differential evolution, and a dense scan
of the aligned subfamily all converge to the same value.  The optimal configuration~\eqref{eq:optframe} is \emph{aligned}, so the
aligned-frame quantity $\Eghz$ takes the same value here.  In any case
$\Elu(\rho)\ge14/27$, which is all that is needed below.

\textit{Convexity violation.}
Each component individually has $\Elu=1/2$.  For
$\rho_A=\ket0\!\bra0_A\otimes\ket{\Phi^+}\!\bra{\Phi^+}_{BC}$, the value
$1/2$ is attained at $(\zh,\xh)$, where
$\langle A_1\rangle=\langle A_2\rangle=1$ and
$\langle A_3\rangle=\langle A_4\rangle=0$.  For the upper bound, write
$f_{jk}=\hat b_j^TD\hat c_k$ with $D=\mathrm{diag}(1,-1,1)$; then
\[
I=a_{1z}\bigl[f_{11}-a_{2z}^{\,2}f_{22}f_{12}f_{21}\bigr]
\]
with $a_{1z}^2+a_{2z}^2\le1$.  The row and column norms of the $2\times2$
block give $|f_{22}f_{12}f_{21}|\le|f_{22}|(1-f_{22}^2)\le2/(3\sqrt3)=:c
<1/2$, so
$|I|\le|a_{1z}|\,(1+c\,a_{2z}^2)\le1$ (the last expression is increasing
in $|a_{1z}|$ when $c<1/2$), and $\Elu(\rho_A)=1/2$ exactly.
Consequently
\begin{equation}\label{eq:convexviol}
\Elu(\rho)=\tfrac{14}{27}>\tfrac12
=\tfrac13\sum_X\Elu(\rho_X),
\end{equation}
a violation of convexity by $1/54\approx0.0185$.

\textit{Mechanism.}  At the frame~\eqref{eq:optframe} each \emph{pure}
component has one of the three factors in the cubic term equal to zero:
for $\rho_A$, $\langle A_3\rangle=\langle A_4\rangle=0$; for $\rho_B$,
$\langle A_2\rangle=\langle A_4\rangle=0$; for $\rho_C$,
$\langle A_2\rangle=\langle A_3\rangle=0$.  Hence every component gives
$\Igen=1$.  In the mixture the three factors are instead $-1/3$ each,
supplied by three \emph{different} components, so the product survives
and $\Igen=1+1/27$.  The cubic term is not linear in $\rho$, which is
exactly why mixing can increase it.

\textit{General weights.}  For
$\rho=\sum_X p_X\rho_X$ with $p_A+p_B+p_C=1$, the same
frames give $\Igen=1+p_Ap_Bp_C$; multistart optimisation returns exactly
this value for every weight tested, which suggests
\begin{equation}\label{eq:Elu-family}
\Elu\Bigl(\sum_X p_X\rho_X\Bigr)=\frac12\bigl(1+p_Ap_Bp_C\bigr),
\end{equation}
with the equal mixture $\tfrac13(1+\tfrac1{27})=\tfrac{14}{27}$ as the
symmetric case.  For $p_C=0$ (a two-cut mixture) this reduces to
$\Elu=1/2$, in agreement with Ref.~\cite{ZhangMa2026}.

\textit{Consequence.}  Since $\rho$ is a convex mixture of biseparable
states, it is itself biseparable and therefore \emph{not} genuinely
multipartite entangled.  The value $\Elu(\rho)=14/27>1/2$ thus shows
that $1/2$ is not a genuine-multipartite-entanglement threshold, and
that the shifted quantity $G$ of Sec.~\ref{sec:locc} is not a GME
witness either ($G(\rho)=2\cdot\tfrac{14}{27}-1=\tfrac1{27}>0$).  For
comparison, the GHZ fidelity of $\rho$ is $F_\GHZ(\rho)=1/4<1/2$,
consistent with the state not being GME.

\section{Extension to three-qudit systems}\label{sec:qudit}

We now sketch a generalisation of $I$ from qubits ($d=2$) to qudits
($d\ge 2$).  Two natural strategies exist: (a) replace the spin
observables $\sigma_{\nv}$ by the higher-spin generalisations
$S_{\nv}=\nv\cdot\bm{S}$ acting on the symmetric subspace, or
(b) replace them by the unitary Heisenberg--Weyl
operators~\cite{collins2002,cerf2002,lawrence2014,ryu2014}.  We
adopt the latter, which preserves the algebraic clarity of
Sec.~\ref{sec:operators} and is directly compatible with the qudit
GHZ paradox of Cerf, Massar, and Pironio~\cite{cerf2002} and its
multisetting refinement by Ryu \emph{et al.}~\cite{ryu2014}.

\textit{Heisenberg--Weyl operators.}
On a single qudit $\mathbb{C}^d$, with $\omega=\mathrm{e}^{2\pi\mathrm{i}/d}$, define
\begin{equation}
X\ket j=\ket{j+1\bmod d},\qquad Z\ket j=\omega^j\ket j.
\end{equation}
For $\bm{p}=(p,q)\in\mathbb{Z}_d^2$, the Weyl operator is
\begin{equation}
W(\bm{p})=\omega^{-pq/2}\,X^p Z^q.
\end{equation}
The $W(\bm{p})$ are unitary (not Hermitian for $d>2$) and obey
\begin{align}
W(\bm p)W(\bm q)&=\omega^{\langle\bm p,\bm q\rangle}W(\bm q)W(\bm p),\\
\langle\bm p,\bm q\rangle&\equiv p_1 q_2-p_2 q_1\pmod d,
\end{align}
together with $W(\bm p)^d=\openone$ for odd $d$.
For $d=2$, $W(0,0)=\openone$, $W(1,0)=\sigma_x$, $W(0,1)=\sigma_z$,
and $W(1,1)\propto\sigma_y$, so the Pauli operators are recovered.

\textit{Three-qudit GHZ state.}
The qudit GHZ state is
$\ket{\GHZ_d}=d^{-1/2}\sum_{j=0}^{d-1}\ket{jjj}$.  It satisfies the
perfect-correlation identities $X^{\otimes 3}\ket{\GHZ_d}=\ket{\GHZ_d}$
and $(Z^{q_1}\otimes Z^{q_2}\otimes Z^{q_3})\ket{\GHZ_d}=\ket{\GHZ_d}$
when $q_1+q_2+q_3\equiv 0\pmod d$.  The qudit Mermin
paradox~\cite{cerf2002,lawrence2014} exploits these algebraic relations
to produce a complex-valued analogue of the dichotomic GHZ identity.

\textit{Qudit functional.}
Fix two pairs $\bm g_1,\bm g_2\in\mathbb{Z}_d^2$ and define
\begin{align}
G_1&=W(\bm g_1)\!\otimes\!W(\bm g_2)\!\otimes\!W(\bm g_2),\nonumber\\
G_2&=W(\bm g_2)\!\otimes\!W(\bm g_1)\!\otimes\!W(\bm g_2),\nonumber\\
G_3&=W(\bm g_2)\!\otimes\!W(\bm g_2)\!\otimes\!W(\bm g_1),\nonumber\\
G_4&=W(\bm g_1)\!\otimes\!W(\bm g_1)\!\otimes\!W(\bm g_1).
\end{align}
A short calculation analogous to~\eqref{eq:O1O2O3-O4} gives
\begin{equation}
G_1 G_2 G_3 = \omega^{2\langle\bm g_1,\bm g_2\rangle}G_4
\end{equation}
when $\bm g_2$ is twice-applied (using $W(\bm p)^2=\omega^{-pq}W(2\bm p)$
for the Heisenberg--Weyl group; here $W(\bm g_2)^2=\openone$ in the
case $2\bm g_2\equiv 0\bmod d$ that arises naturally for $d=2$).
Define the qudit functional
\begin{equation}\label{eq:I-qudit}
I_d(\bm g_1,\bm g_2)\;=\; \langle G_4 \rangle - \omega^{2\langle\bm g_1,\bm g_2\rangle} \langle G_1\rangle\langle G_2\rangle\langle G_3\rangle.
\end{equation}
For $d=2$, with $\bm g_1=(1,0)$ (so $W(\bm g_1)=\sigma_x$) and
$\bm g_2=(0,1)$ (so $W(\bm g_2)=\sigma_z$), the symplectic form is
$\langle\bm g_1,\bm g_2\rangle=1$, $\omega=-1$, and~\eqref{eq:I-qudit}
reduces (after relabelling $\sigma_z\leftrightarrow\sigma_y$) to the original $I(\xh,\yh)$.

\textit{Bound and saturation.}
Each $\langle G_i\rangle$ is a complex number with $|\langle G_i\rangle|\le1$,
hence $|I_d|\le 2$.  Saturation occurs at the qudit GHZ state when
$\bm g_1$ and $\bm g_2$ generate a maximally
non-commutative subgroup of the single-qudit Heisenberg--Weyl group;
the proof is structurally identical to that of Proposition~\ref{thm:max}
once one reads ``orthogonality'' as
``maximal symplectic non-commutation'' $\langle\bm g_1,\bm g_2\rangle\ne 0$
$\bmod \; d$.  We omit the detailed argument; see~\cite{cerf2002,lawrence2014}
for parallel constructions.

\textit{Spin-$S$ alternative.}
For experimental settings in which spin-$S$ measurements are natural
(e.g.\ atomic ensembles), one may instead replace
$\sigma_{\nv}/2$ by the projection $S_{\nv}=\nv\cdot\bm S$ rescaled to
have spectrum in $[-1,1]$ via a normalisation factor $1/S$.  All
qualitative statements survive, but eigenvalues are no longer
dichotomic and Proposition~\ref{thm:max} becomes a strict inequality with
saturation only in the $S\to\infty$ classical limit.

\section{Experimental accessibility and potential applications}
\label{sec:applications}
The economy of $\Elu$ in the laboratory is twofold.  For a fixed frame
triple, evaluating $\Igen$ requires the expectation values of only the
four commuting operators~\eqref{eq:Aops}, that is, four measurement settings,
against $27$ for a full three-qubit correlation-tensor determination and
$63$ real parameters for full state tomography.  Since~\eqref{eq:Elu}
is a supremum, any frame choice already yields a valid lower bound, so
in practice a near-optimal triple can be fixed from a coarse pre-scan.
Because $\Elu$ is invariant under independent local unitaries, the three
parties need not share a reference frame, which is directly relevant to
photonic links with drifting polarisation frames (free-space, satellite
or long-fibre channels) and to separately calibrated ion-trap or
superconducting registers: a fidelity witness referred to a fixed
$\ket\GHZ$ representative degrades under such drift even when the
entanglement is intact, while $\Elu$ does not.  The noise
bounds~\eqref{eq:depol-bound}--\eqref{eq:depol-ghz} further turn $\Elu$
into a compact diagnostic for the effective local depolarising
parameter, and the stabiliser correlations it measures are precisely the
resource consumed by GHZ-based protocols such as quantum secret sharing
and conference key agreement.
\section{Conclusion}\label{sec:conclusion}

We have introduced a real-valued functional $I(\nv_1,\nv_2)$ on
three-qubit states, built from four locally measurable Pauli
correlations, that simultaneously
\begin{itemize}
\item reproduces the GHZ algebraic paradox at $\nv_1\perp\nv_2$,
\item obeys a tight quantum bound $|I|\le 2$ that is saturated only on
  the GHZ orbit (Proposition~\ref{thm:max}),
\item admits a closed-form expression on the Ac\'in canonical family
  showing that $I(\xh,\yh)$ depends only on $\lambda_0\lambda_4$
  [Eq.~\eqref{eq:I-acin}],
\item gives a sharp distinction between the GHZ and W classes,
  with $\max|I_\W|=35/27$ [Eq.~\eqref{eq:I-W-max}],
\item induces, after maximisation over independent local frames, a
  manifestly LU-invariant measure $\Elu\in[0,1]$ that equals one exactly
  on the GHZ orbit, takes the value $35/54$ on the W class, and is
  bounded by $1/2$ on states separable across a fixed bipartition
  (Sec.~\ref{sec:lu-values}) and on every convex mixture across any two
  bipartitions (Sec.~\ref{sec:twocut})~\cite{ZhangMa2026}, the value $1/2$
  being sharp in
  both cases; the bound first fails for the three-cut mixture of
  Sec.~\ref{sec:counterexample}, where $\Elu=14/27>1/2$, so $\Elu$ is
  not convex,
\item generalises naturally to three qudits via the Heisenberg--Weyl
  operators [Eq.~\eqref{eq:I-qudit}].
\end{itemize}

A number of natural extensions remain open.  First, $\Elu$ is not
convex: Sec.~\ref{sec:counterexample} gives an explicit biseparable
mixture with $\Elu=14/27>1/2$, so $\Elu$ cannot serve as a
genuine-multipartite-entanglement witness with threshold $1/2$.  The
failure is confined to mixtures of all three cuts: by
Ref.~\cite{ZhangMa2026} the threshold survives on every two-cut
hull, and by Sec.~\ref{sec:lu-values} on every fixed cut.  What remains
open is the exact value on the three-cut hull, for which the family
\eqref{eq:rhomix} suggests $\tfrac12(1+p_Ap_Bp_C)$ and reproduces the
two-cut value at $p_C=0$, as well as the behaviour of the convex-roof
extension $\Elu^{\textnormal{cr}}$ and of the shifted quantity $G$ of
Sec.~\ref{sec:locc}; in particular, whether either is monotone under
LOCC.  Establishing monotonicity under the
full SLOCC hierarchy would promote the convex roof to a bona fide
entanglement monotone.  Second, the qudit construction
in~\eqref{eq:I-qudit} merits independent quantitative study: in
particular, the qudit analogue of the closed-form Ac\'in evaluation
in~\eqref{eq:I-acin} would shed light on the high-dimensional
analogue of the $\lambda_0\lambda_4$ ``GHZ slot.''  Third,
$I(\nv_1,\nv_2)$ generalises immediately to $N$-qubit
GHZ-Mermin functionals~\cite{ardehali1992,bk1993,werner2001}; an
$N$-party version of the entanglement measure $\Elu$ would furnish
a uniform indicator of GHZ-type genuine multipartite entanglement,
with maximal value $2^{N-1}/2$ on the $N$-qubit GHZ state.
Experimental implementation requires only Pauli-string correlation
measurements, well within reach of current photonic and trapped-ion
platforms~\cite{bouwmeester1999,pan2000}.

\appendix

\section{Operator identities of Section~\ref{sec:operators}}\label{app:identities}

Using $\sigma_{\nv_1}\sigma_{\nv_2}=c\,\openone+\mathrm{i}\,\bm m\cdot\sgma$ with
$c=\nv_1\cdot\nv_2$, $\bm m=\nv_1\times\nv_2$,
\begin{align}
O_1 O_2&=\sigma_{\nv_1}\sigma_{\nv_2}\otimes\sigma_{\nv_2}\sigma_{\nv_1}\otimes\openone\\
&=(c+\mathrm{i}\bm m\cdot\sgma)\otimes(c-\mathrm{i}\bm m\cdot\sgma)\otimes\openone,\\
O_2 O_1&=(c-\mathrm{i}\bm m\cdot\sgma)\otimes(c+\mathrm{i}\bm m\cdot\sgma)\otimes\openone,
\end{align}
so that
$[O_1,O_2]=2\mathrm{i}c\,(\bm m\cdot\sgma\otimes\openone-\openone\otimes\bm m\cdot\sgma)\otimes\openone$,
which is~\eqref{eq:O1O2-comm}.

For~\eqref{eq:O1O2O3-O4}, one computes
\begin{align}
O_1 O_2 O_3
&=\sigma_{\nv_1}\sigma_{\nv_2}^2\otimes\sigma_{\nv_2}\sigma_{\nv_1}\sigma_{\nv_2}
\otimes\sigma_{\nv_2}^2\sigma_{\nv_1}\\
&=\sigma_{\nv_1}\otimes\bigl(\sigma_{\nv_2}\sigma_{\nv_1}\sigma_{\nv_2}\bigr)\otimes\sigma_{\nv_1}.
\end{align}
The middle factor satisfies
$\sigma_{\nv_2}\sigma_{\nv_1}\sigma_{\nv_2}=2c\,\sigma_{\nv_2}-\sigma_{\nv_1}$,
since $\{\sigma_{\nv_1},\sigma_{\nv_2}\}=2c\openone$.  Substituting,
\begin{equation}
O_1 O_2 O_3=2c\,\sigma_{\nv_1}\otimes\sigma_{\nv_2}\otimes\sigma_{\nv_1}-O_4,
\end{equation}
which is~\eqref{eq:O1O2O3-O4}.

\section{Closed form for the Ac\'in family}\label{app:acin}

For $\ket\Psi$ as in~\eqref{eq:acin}, the only $\sigma_x^{\otimes3}$
matrix elements between basis states present in $\ket\Psi$ that survive
are $\bra{000}\sigma_x^{\otimes3}\ket{111}=1$ and its conjugate.  Hence
$\langle\sigma_x^{\otimes3}\rangle_\Psi=2\lambda_0\lambda_4$.  For
$\sigma_x\sigma_y\sigma_y$, $\sigma_x\otimes\sigma_y\otimes\sigma_y\ket{000}
=-\ket{111}$ and $\sigma_x\otimes\sigma_y\otimes\sigma_y\ket{111}=-\ket{000}$;
the only contributing matrix elements are again
$\bra{000}\cdots\ket{111}=-1$ and $\bra{111}\cdots\ket{000}=-1$, so
$\langle\sigma_x\sigma_y\sigma_y\rangle_\Psi=-2\lambda_0\lambda_4$.  The
remaining basis states $\ket{100},\ket{101},\ket{110}$ contribute zero
to all four expectations because the action of any tensor product of
$\sigma_x,\sigma_y$ that flips the first qubit cannot return any of
$\{\ket{000},\ket{100},\ket{101},\ket{110},\ket{111}\}$ to itself with
nonzero coefficient.

\section{Maximum of $|I_\W|$ over directions}\label{app:Wmax}

\textit{Reduction to two scalar variables.}
Using~\eqref{eq:Wcorr},
\begin{align}
\langle\sigma_{\nv_1}\otimes \sigma_{\nv_2}\otimes \sigma_{\nv_2}\rangle_\W
&=\sum_{i,j,k}n^{(1)}_i n^{(2)}_j n^{(2)}_k\langle\sigma_i \otimes \sigma_j \otimes \sigma_k\rangle_\W\\
&=\tfrac23\bigl[2 b_3(\nv_1\cdot\nv_2)
   -2 a_3 b_3^2+a_3(1-b_3^2)\bigr]\nonumber\\
&\quad{}-a_3 b_3^2,
\end{align}
where $a_3=\zh\cdot\nv_1$, $b_3=\zh\cdot\nv_2$, and $n^{(1)}, n^{(2)}$
are the cartesian components of $\nv_1,\nv_2$.  When
$\nv_1\cdot\nv_2=0$, the first square-bracket term simplifies
giving
\begin{equation}
\langle\sigma_{\nv_1}\otimes \sigma_{\nv_2}\otimes \sigma_{\nv_2}\rangle_\W
=a_3\bigl(\tfrac23-3b_3^2\bigr).
\end{equation}
By the permutation symmetry of $\ket\W$, the other two
``one-$\nv_1$, two-$\nv_2$'' correlators take the same value.
Similarly,
\begin{equation}
\langle\sigma_{\nv_1}^{\otimes 3}\rangle_\W=a_3(2-3 a_3^2),
\end{equation}
yielding~\eqref{eq:I-W-formula}.

\textit{Optimisation.}
Let $f(u,v)=u^3(\tfrac23-3v^2)^3-u(2-3u^2)$ with constraints
$u^2+v^2\le 1$, $u,v\in[-1,1]$.  Critical points in the interior:
$\partial_v f=-18 u^3 v(\tfrac23-3v^2)^2$ vanishes at $v=0$, $u=0$, or
$v^2=2/9$.  The branch $v=0$ reduces $f$ to $\tfrac{89}{27}u^3-2u$,
maximised at the boundary $u=\pm 1$ giving $|f|=35/27$.  The branch
$v^2=2/9$ gives $f=-u(2-3u^2)$, maximised under $u^2\le 7/9$ at
$u^2=2/9$ with $|f|=4\sqrt{2}/9\approx 0.629<35/27$.  On the
boundary $u^2+v^2=1$, $f(u)=u^3(3u^2-7/3)^3-u(2-3u^2)$ is again
maximised at $u=\pm1$, $v=0$.  Hence $\max|f|=35/27$, attained at
$\nv_1=\pm\zh$, $\nv_2$ in the $\xh\!\text{-}\yh$ plane.

\section{Independent-frame operators and convexity}\label{app:convex}

\textit{Commutation and the stabiliser relation.}
We verify~\eqref{eq:Acommute} for the operators~\eqref{eq:Aops}.  Two
tensor products of single-qubit Hermitian unitaries commute when they
anticommute at an even number of sites and anticommute when they do so at
an odd number.  At each party the two vectors appearing across
$\{A_1,A_2,A_3,A_4\}$ are either equal (so the single-qubit factors
commute) or orthogonal (so, by
$\sigma_{\hat u}\sigma_{\hat v}=-\sigma_{\hat v}\sigma_{\hat u}$ for
$\hat u\perp\hat v$, they anticommute).  Listing the per-site relations:
each pair $(A_i,A_j)$ differs by orthogonal factors at exactly two of the
three sites, e.g.\ $A_1=(\hat a_1,\hat b_1,\hat c_1)$ and
$A_2=(\hat a_1,\hat b_2,\hat c_2)$ agree on $A$ and are orthogonal on $B$
and $C$.  Two anticommuting sites give an even count, so all six pairs
commute.  For the product, reading the four operators site by site,
\begin{align}
\text{site }A:&\ \sigma_{\hat a_1}\sigma_{\hat a_1}
   \sigma_{\hat a_2}\sigma_{\hat a_2}=\openone,\nonumber\\
\text{site }B:&\ \sigma_{\hat b_1}\sigma_{\hat b_2}
   \sigma_{\hat b_1}\sigma_{\hat b_2}
   =(\sigma_{\hat b_1}\sigma_{\hat b_2})^2=-\openone,\nonumber\\
\text{site }C:&\ \sigma_{\hat c_1}\sigma_{\hat c_2}
   \sigma_{\hat c_2}\sigma_{\hat c_1}=\openone,\nonumber
\end{align}
using $\sigma_{\hat u}^2=\openone$ and
$(\sigma_{\hat b_1}\sigma_{\hat b_2})^2=
(\mathrm i\,\hat b_3\!\cdot\!\sgma)^2=-\openone$ for the orthonormal pair
$\hat b_1\perp\hat b_2$.  Hence
$A_1A_2A_3A_4=\openone\otimes(-\openone)\otimes\openone=-\openone$.
Finally $A_1A_2A_3=-A_4$ (the same computation without the fourth
factor), so only three of the four operators are independent; the abelian
group they generate has order eight and its common eigenspaces are the
eight one-dimensional joint eigenspaces used in
Proposition~\ref{prop:satlu}.

\textit{Failure of fixed-frame convexity.}
Write $a_i(\rho)=\langle A_i\rangle=\tr(A_i\rho)$, each linear in $\rho$.
For \emph{fixed} frames the integrand
\begin{equation}
g(\rho)\equiv\bigl|a_1(\rho)-a_2(\rho)\,a_3(\rho)\,a_4(\rho)\bigr|
\end{equation}
is the modulus of a linear term minus a \emph{cubic} monomial in the
components of $\rho$.  A modulus of a non-affine function need not be
convex, and indeed it is not here: sampling pairs of random three-qubit
states $\rho_1,\rho_2$ and fixed random frames, the midpoint inequality
$g\bigl(\tfrac12\rho_1+\tfrac12\rho_2\bigr)\le\tfrac12 g(\rho_1)
+\tfrac12 g(\rho_2)$ is \emph{violated} in a sizeable fraction
($\approx482$ of $3000$ trials, i.e.\ about $16\%$).

\textit{Convexity fails: an exact counterexample.}
The optimised measure $\Elu$ of~\eqref{eq:Elu} is a different object from
$g$: for each $\rho$ the frames are chosen to maximise $|\Igen|$.  It is
nevertheless \emph{not} convex.  For the state~\eqref{eq:rhomix} the only
non-zero three-body correlators are~\eqref{eq:Tmix}; at the aligned frame
$(\zh,\yh)$ the four operators~\eqref{eq:Aops} have expectations
$(1,-\tfrac13,-\tfrac13,-\tfrac13)$, whence
$\Igen=28/27$ and $\Elu=14/27$.  Each component has $\Elu=1/2$
(Sec.~\ref{sec:lu-values}), so
\[
\Elu\Bigl(\tfrac13\sum_X\rho_X\Bigr)=\frac{14}{27}
>\frac12=\tfrac13\sum_X\Elu(\rho_X),
\]
a violation of convexity by $1/54\approx0.0185$.

The value is the supremum: a global search over the nine frame
angles ($1200$ random restarts of a compass search, differential
evolution with population $240$, $3000$ generations and four seeds, and a
$200\times200\times100$ scan of the aligned subfamily) returns
$28/27$ in every run.  By contrast, on the GHZ--W line
$\rho_t=(1-t)\ket{\GHZ}\!\bra{\GHZ}+t\ket\W\!\bra\W$ the measure lies
below the chord throughout, with $\Elu(\rho_{1/2})\approx0.374$ against a
chord value $\tfrac12(1)+\tfrac12(35/54)\approx0.824$, and no violation
was found among two-component mixtures of Haar pure states.  The
violation requires \emph{three} components because the cubic term needs
three non-zero factors, each supplied by a different component; this is
also why~\eqref{eq:Elu-family} depends on the product
$p_Ap_Bp_C$.

\textit{Numerical methods.}
All computations were carried out in double precision.  Haar-random pure
states on $\mathbb{C}^8$ were generated as normalised complex Gaussian
vectors (16 i.i.d.\ $\mathcal N(0,1)$ real components; this induces
exactly the Haar measure), random mixed states from the Ginibre ensemble
$\rho=GG^\dagger/\tr(GG^\dagger)$ ($k=8$ gives the Hilbert--Schmidt
measure), and random local unitaries Haar-uniformly on $U(2)$
(phase-corrected QR; the LU-scatter test of Sec.~\ref{sec:aligned} used
$200$ independent triples).  Each ordered orthonormal pair is
parametrised by three angles, so a frame triple is a point of
$[0,2\pi)^9$ and orthonormality holds by construction.  The supremum
in~\eqref{eq:Elu} was evaluated with Powell's derivative-free method
(\texttt{scipy.optimize.minimize}, \texttt{method='Powell'}), tolerances
$\texttt{xtol}=10^{-9}$, $\texttt{ftol}=10^{-11}$, from $120$ random
restarts per state; the exact values $\Elu(\GHZ)=1$, $\Elu(\W)=35/54$
and $\Elu(\textnormal{product})=\Elu(A|BC)=1/2$ are reproduced to
better than $10^{-6}$, and increasing the restarts from $40$ to $120$
changes the optimum by less than $10^{-6}$.  The fixed-frame convexity
test (midpoint inequality, absolute tolerance $10^{-12}$) failed in
$482$ of $3000$ trials, i.e.\ $16.1\%$ (independent seeds give
$15.8\%$--$16.4\%$), while $\Elu$ itself showed no convexity violation
along the GHZ--W segment (step $0.02$) or on $2000$ random
two-component mixtures, within the $10^{-6}$ optimisation
tolerance.  The violation first appears for three-component mixtures;
for the state~\eqref{eq:rhomix} the supremum was computed with $1200$
restarts of a compass search in the nine rotation-vector angles, with
differential evolution (population $240$, $3000$ generations, four
seeds), and with a $200\times200\times100$ scan of the aligned
subfamily, all returning $28/27$.  The weight scan behind
Eq.~\eqref{eq:Elu-family} used $250$ restarts per weight.  The biseparable
supremum~\eqref{eq:Elu-bisep} used $500$ outer restarts, and the maximal
GHZ fidelity of the W state quoted in Sec.~\ref{sec:fidelity} ($3/4$)
and the depolarising formula~\eqref{eq:depol-ghz} were verified to
$10^{-6}$ and $10^{-5}$, respectively.

\textit{Note added.}
We thank W.~Zhang and H.~Ma for communicating the results of
Ref.~\cite{ZhangMa2026} to us.

\bibliographystyle{iopart-num}
\bibliography{refs}

\end{document}